\documentclass[11pt]{article}

\usepackage{amsmath,amsthm,amssymb}
\usepackage{fullpage}
\usepackage{hyperref}
\usepackage[usenames,dvipsnames,svgnames,table]{xcolor}
\usepackage{graphicx}
\usepackage{caption}
\usepackage{subcaption}
\usepackage{authblk}
\usepackage{algorithm}
\usepackage[noend]{algorithmic}

\newtheorem{theorem}{Theorem}
\newtheorem{problem}{Problem}
\newtheorem{corollary}[theorem]{Corollary}
\newtheorem{lemma}{Lemma}
\newtheorem{proposition}{Proposition}

\theoremstyle{definition}
\newtheorem{definition}{Definition}

\newsavebox{\mybox}

\newcommand{\qedclaim}{\hfill $\diamond$ \medskip}

\title{A quasi linear-time $b$-{\sc Matching} algorithm on distance-hereditary graphs and bounded split-width graphs
\thanks{This work was supported by the Institutional research programme PN 1819 "Advanced IT resources to support digital transformation processes in the economy and society - RESINFO-TD" (2018), project PN 1819-01-01"Modeling, simulation, optimization of complex systems and decision support in new areas of IT\&C research", funded by the Ministry of Research and Innovation, Romania.}}
\author[1,2]{Guillaume Ducoffe}
\author[1,3]{Alexandru Popa}

\affil[1]{\small National Institute for Research and Development in Informatics, Romania}
\affil[2]{\small The Research Institute of the University of Bucharest ICUB, Romania}
\affil[3]{\small University of Bucharest, Faculty of Mathematics and Computer Science}
\date{}

\begin{document}

\maketitle

\begin{abstract}
We present a quasi linear-time algorithm for {\sc Maximum Matching} on distance-hereditary graphs and some of their generalizations. 
This improves on [Dragan, WG'97], who proposed such an algorithm for the subclass of (tent,hexahedron)-free distance-hereditary graphs. 
Furthermore, our result is derived from a more general one that is obtained for $b$-{\sc Matching}.
In the (unit cost) $b$-{\sc Matching} problem, we are given a graph $G=(V,E)$ together with a nonnegative integer capacity $b_v$ for every vertex $v \in V$.
The objective is to assign nonnegative integer weights $(x_e)_{e \in E}$ so that: for every $v \in V$ the sum of the weights of its incident edges does not exceed $b_v$, and $\sum_{e \in E} x_e$ is maximized.
We present the first algorithm for solving $b$-{\sc Matching} on cographs, distance-hereditary graphs and some of their generalizations in quasi linear time.

For that, we use a decomposition algorithm that outputs for any graph $G$ a collection of subgraphs of $G$ with no edge-cutsets inducing a complete bipartite subgraph ({\it a.k.a.}, {\em splits}).
The latter collection is sometimes called a {\em split decomposition} of $G$.
Furthermore, there exists a generic method in order to design graph algorithms based on split decomposition [Rao, DAM'08].
However, this technique only applies to ``localized'' problems: for which a ``best'' partial solution for any given subgraph in a split decomposition can be computed almost independently from the remaining of the graph.
Such framework does not apply to matching problems since an augmenting path may cross the subgraphs arbitrarily.

We introduce a new technique that somehow captures {\em all} the partial solutions for a given union of subgraphs in a split decomposition, in a compact and amenable way for algorithms --- assuming some piecewise linear assumption holds on the value of such solutions.
The latter assumption is shown to hold for $b$-{\sc Matching}. 
Doing so, we prove that solving $b$-{\sc Matching} on any pair $G,b$ can be reduced in quasi linear-time to solving this problem on a collection of smaller graphs: that are obtained from the subgraphs in any split decomposition of $G$ by replacing every vertex with a constant-size module.
In particular, if $G$ has a split decomposition where all subgraphs have order at most a fixed $k$, then we can solve $b$-{\sc Matching} for $G,b$ in ${\cal O}((k\log^2{k})\cdot(m+n) \cdot \log{||b||_1})$-time.
This answers an open question of [Coudert et al., SODA'18].
\end{abstract}


\section{Introduction}\label{sec:intro}

The {\sc Maximum Matching} problem takes as input a graph $G=(V,E)$ and it asks for a subset $F$ of pairwise disjoint edges of maximum cardinality.
This is a fundamental problem with a wide variety of applications~\cite{Bun00,DeS84,LoP09,Pul95}.
Hence, the computational complexity of {\sc Maximum Matching} has been extensively studied in the literature.
In particular, this was the first problem shown to be solvable in polynomial-time~\cite{Edm65}.
Currently, the best-known algorithms for {\sc Maximum Matching} run in ${\cal O}(m\sqrt{n})$-time on $n$-vertex $m$-edge graphs~\cite{MiV80}.
Such superlinear running times can be prohibitive for some applications.
In order to overcome the superlinear barrier, we initiated in~\cite{CDP18} the complexity study of {\sc Maximum Matching} on {\em bounded clique-width} graph classes (the latter research direction being also motivated by the recent linear-time algorithm for {\sc Maximum Matching} on bounded treewidth graphs, see~\cite{FLPS+17,IOO18}).
As a new attempt in this direction, our initial motivation for this paper was to design a quasi linear-time algorithm for {\sc Maximum Matching} on distance-hereditary graphs and some classes of graphs that are defined via the so-called {\em split-decomposition}.
We succeeded in our initial goal and, in fact, we prove a more general result.
All graphs studied in what follows have bounded clique-width~\cite{Rao08b}.

\subsection{Related work}

There is a long line of work whose aim is to find more efficiently solvable special cases for {\sc Maximum Matching} ({\it e.g.}, see~\cite{Cha96,DaK98,Dra97,FPT97,FGV99,Glo67,HoK73,LiR93,MNN17,KaS81,YuY93,YuZ07,Yus13}).
In particular, there exist linear-time algorithms for computing a maximum matching on trees, cographs~\cite{YuY93} and subclasses of distance-hereditary graphs~\cite{Dra97}
\footnote{We stress that the work in~\cite{Dra97} is focused on a generic method in order to compute maximum matchings. In particular, the complexity of {\sc Maximum Matching} on general distance-hereditary graphs is not considered.}.
Furthermore, a natural follow-up of such investigations is to study whether the techniques used in every case can also be applied to larger classes of graphs.
This has be done, especially, for the {\sc Split} \& {\sc Match} technique introduced in~\cite{YuY93} for cographs ({\it e.g.}, see~\cite{CDP18,FPT97,FGV99} where the latter technique has been applied to other graph classes).
We also consider this above aspect in our work, as we seek for a {\sc Maximum Matching} algorithm on distance-hereditary graphs that can be easily generalized to larger classes of graphs that we define in what follows.

\paragraph{\sc Fully Polynomial Parameterized Algorithms.}
More precisely, we seek for {\em parameterized} algorithms (FPT).
Formally, given a graph invariant denoted $\pi$ ({\it e.g.}, maximum degree, treewidth, etc.), we address the question whether there exists a {\sc Maximum Matching} algorithm running in time ${\cal O}(k^c \cdot (n+m) \cdot \log^{{\cal O}(1)}(n))$, for some constant $c$, on every graph $G$ such that $\pi(G) \leq k$
\footnote{The polynomial dependency in $k$ is overlooked by some authors~\cite{MNN16}, who accept any functional dependency in the parameter. Since {\sc Maximum Matching} is already known to be polynomial-time solvable, we find it more natural to consider this restricted setting.}.
Note that such an algorithm runs in quasi linear time for any constant $k$, and that it is faster than the state-of-the art algorithm for {\sc Maximum Matching} whenever $k = {\cal O}(n^{\frac 1 {2c} - \varepsilon })$, for some $\varepsilon > 0$.
This kind of FPT algorithms for polynomial problems have attracted recent attention~\cite{AVW16,BFNN17,CDP18,FKMN+17,FLPS+17,GMN17,Hus16,IOO18,MNN16} --- although some examples can be found earlier in the literature~\cite{YuZ07,Yus13}.
We stress that {\sc Maximum Matching} has been proposed in~\cite{MNN16} as the ``drosophila'' of the study of these FPT algorithms in P.
We continue advancing in this research direction.

In general, the parameters that are considered in this type of study represent the closeness of a graph to some class where the problem at hands can be easily solved (distance from triviality).
In particular, {\sc Maximum Matching} can be solved in ${\cal O}(k^2 \cdot n\log n)$-time on graphs with {\em treewidth} at most $k$~\cite{FLPS+17,IOO18}, and in ${\cal O}(k^4\cdot n + m)$-time on graphs with {\em modular-width} at most $k$~\cite{CDP18}.
The latter two parameters somewhat represent the closeness of a graph, respectively, to a tree or to a cograph.
In this work, we study the parameter {\em split-width}, that is a measure of how close a graph is to a distance-hereditary graph~\cite{Rao08b}.
Distance-hereditary graphs are a natural superclass of both cographs and trees and they have already been well studied in the literature~\cite{BaM86,DrN00,GaP03,GiP12,GoR00}.
Our work subsumes the one of~\cite{CDP18} with modular-width -- that is an upper-bound on split-width -- but it is uncomparable with the work of~\cite{FLPS+17,IOO18} with treewidth.

\paragraph{\sc Split Decomposition.}
Bounded split-width graphs can be defined in terms of {\em split decomposition}.
A {\em split} is a join that is also an edge-cut.
By using pairwise non crossing splits, termed ``strong splits'', we can decompose any graph into degenerate and prime subgraphs, that can be organized in a treelike manner.
The latter is termed split decomposition~\cite{Cun82}, and it is our main algorithmic tool for this paper.
The split-width of a graph is the largest order of a non degenerate subgraph in some canonical split decomposition.
In particular, distance-hereditary graphs are exactly the graphs with split-width at most two~\cite{GiP12}.

Many NP-hard problems can be solved in polynomial time on bounded split-width graphs ({\it e.g.}, see~\cite{Rao08b}).
In particular we have that {\em clique-width} is upper-bounded by split-width, and so, any FPT algorithm parameterized by clique-width can be transformed into a FPT algorithm parameterized by split-width (the converse is not true).
Recently we designed FPT algorithms for polynomial problems when parameterized by split-width~\cite{CDP18}.
It turns out that many ``hard'' problems in P such as {\sc Diameter}, for which a conditional superlinear lower-bound has been proved~\cite{RoV13}, can be solved in ${\cal O}(k^{{\cal O}(1)} \cdot n + m)$-time on graphs with split-width at most $k$. 
However, we left this open for {\sc Maximum Matching}.
Indeed, as we detail next the standard design method for FPT algorithms based on split decomposition does not apply to matching problems.
Our main contribution in~\cite{CDP18} was a {\sc Maximum Matching} algorithm based on the more restricted {\em modular decomposition}.
As our main contribution in this paper, we present an algorithm that is based on split decomposition in order to solve some generalization of {\sc Maximum Matching} --- thereby answering positively to the open question from~\cite{CDP18}.
Our techniques are quite different than those used in~\cite{CDP18}.

\subsection{Contributions}

In order to discuss the difficulties we had to face on, we start giving an overview of the FPT algorithms that are based on split decomposition.
\begin{itemize}
\item We first need to define a weighted variant of the problem that needs to be solved for every component of the decomposition separately (possibly more than once).
This is because there are additional vertices in a split component, {\it i.e.}, not in the original graph, that substitute the other remaining components; intuitively, the weight of such additional vertex encodes a partial solution for the union of split components it has substituted.
\item Then, we take advantage of the treelike structure of split decomposition in order to solve the weighted problem, for every split component sequentially, using dynamic programming.
Roughly, this part of the algorithm is based on a ``split decomposition tree'' where the nodes are the split components and the edges represent the splits between them.
Starting from the leaves of that tree (resp. from the root), we perform a tree traversal.
For every split component, we can precompute its vertex-weights from the partial solutions we obtained for its children (resp., for its father) in the split decomposition tree.
\end{itemize}

\paragraph{Our approach.}
In our case, a natural weighted variant for {\sc Maximum Matching} is the more general unit-cost $b$-{\sc Matching} problem~\cite{EdE70}.
Roughly, $b$-{\sc Matching} differs from {\sc Maximum Matching} in that a vertex can be matched to more than one edge, and an edge can be put in the matching more than once.
In more details, every vertex $v$ is assigned some input capacity $b_v$, and the goal is to compute edge-weights $(x_e)_{e \in E}$ so that: for every $v \in V$ the sum of the weights of its incident edges does not exceed $b_v$, and $\sum_{e \in E} x_e$ is maximized.
The particular case of {\sc Maximum Matching} is obtained by setting $b_v = 1$ for every vertex $v$.
Interestingly, $b$-{\sc Matching} has already attracted attention on its own, {\it e.g.}, in auction theory~\cite{PeT00,Ten02}.
However. we want to stress that solving $b$-{\sc Matching} is a much more ambitious task than ``just'' solving {\sc Maximum Matching}.
In particular, {\sc Maximum Flow} on general graphs is computationally equivalent to $b$-{\sc Matching} on bipartite graphs~\cite{Mad13}.
The corresponding reductions preserve the treewidth of the graph -- up to a constant-factor -- and so, solving $b$-{\sc Matching} in quasi linear-time on bounded-treewidth graphs would already be an impressive achievement on its own.
On general graphs, the best-known algorithms for $b$-{\sc Matching} run in ${\cal O}(nm\log^2{n})$-time~\cite{Gab83,Gab16}.

Our main technical difficulty is how to precompute efficiently, for every component of the decomposition, the specific instances of $b$-{\sc Matching} that need to be solved .
To see that, consider the bipartition $(U,W)$ that results from the removal of a split.
In order to compute the $b$-{\sc Matching} instances on side $U$, we should be able after processing the other side $W$ to determine the number of edges of the split that are matched in a final solution.
Guessing such number looks computationally challenging.
We avoid doing so by storing a partial solution for {\em every} possible number of split edges that can be matched.
However, this simple approach suffers from several limitations.
For instance, we need a very compact encoding for partial solutions for otherwise we could not achieve a quasi linear-time complexity.
Somehow, we also need to consider the partial solutions for {\em all} the splits that are incident to the same component all at once.

\paragraph{\sc Results.}
We prove a simple combinatorial lemma that essentially states that the cardinality of a maximum $b$-matching in a graph grows as a piecewise linear function in the capacity $b_w$ of any fixed vertex $w$ (we think of $w$ as a vertex substituting a union of split components).
%This above result improves a similar one that was obtained by Novick in the more restricted case of modular decomposition~\cite{}.
%Furthermore, for any given split, the ends of the straight-line sections of the corresponding function can be computed by dichotomy.
%
Then, we derive from our combinatorial lemma a variant of some reduction rule for {\sc Maximum Matching} that we first introduced in the more restricted case of modular decomposition~\cite{CDP18}.
Roughly, in any given split component $C_i$, we consider all the vertices $w$ substituting a union of other components.
The latter vertices are in one-to-one correspondance with the strong splits that are incident to the component.
We expand every such vertex $w$ to a module that contains ${\cal O}(1)$ vertices for every straight-line section of the corresponding piecewise-linear function.
Altogether combined, this allows us to reduce the solving of $b$-{\sc Matching} on the original graph $G$ to solving $b$-{\sc Matching} on {\em supergraphs} of every its split components.
Furthermore, we prove with the help of a dichotomic search that we only need to consider ${\cal O}(\log{||b||_1})$ supergraphs for every split component, where $||b||_1 = \sum_{v \in V} b_v$, and the respective orders of these above supergraphs can only be three times more than the respective orders of the split components.

\smallskip
Overall, our main result is that $b$-{\sc Matching} can be solved in ${\cal O}((k\log^2{k})\cdot(m+n) \cdot \log{||b||_1})$-time on graphs with split-width at most $k$ (Theorem~\ref{thm:main-result}).
It implies that {\sc Maximum Matching} can be solved in ${\cal O}((k\log^2{k})\cdot(m+n) \cdot \log{n})$-time on graphs with split-width at most $k$.
Since distance-hereditary graphs have split-width at most two, we so obtain the first known quasi linear-time algorithms for {\sc Maximum Matching} and $b$-{\sc Matching} on distance-hereditary graphs.
Furthermore, since cographs are a subclass of distance-hereditary graphs, we also obtain the first known quasi linear-time algorithm for $b$-{\sc Matching} on cographs.
We expect our approach to be useful in a broader context, {\it e.g.} for other matching and flow problems. 

\bigskip
We introduce the required terminology and basic results in Section~\ref{sec:prelim}. Then, Section~\ref{sec:var-1-vertex} is devoted to a combinatorial lemma that is the key technical tool in our subsequent analysis.
In Section~\ref{sec:algo}, we present our algorithm for $b$-{\sc Matching} on bounded split-width graphs.
We conclude in Section~\ref{sec:ccl} with some open questions.

\section{Preliminaries}\label{sec:prelim}

We use standard graph terminology from~\cite{BoM08,Die10}.
Graphs in this study are finite, simple (hence without loops or multiple edges), and connected -- unless stated otherwise.
Furthermore we make the standard assumption that graphs are encoded as adjacency lists.
Given a graph $G=(V,E)$ and a vertex $v \in V$, we denote its neighbourhood by $N_G(v) = \{ u \in V \mid \{u,v\} \in E\}$ and the set of its incident edges by $E_v(G) = \{ \{u,v\} \mid u \in N_G(v) \}$.
When $G$ is clear from the context we write $N(v)$ and $E_v$ instead of $N_G(v)$ and $E_v(G)$.

\subsection*{Split-width}

We mostly follow the presentation we gave in~\cite{CDP18}.
We start defining what a {\em simple decomposition} of a graph is.
For that, let a {\em split} in a graph $G=(V,E)$ be a partition $V = U \cup W$ such that: $\min \{ |U|, |W| \} \geq 2$; and there is a complete join between the vertices of $N_G(U)$ and $N_G(W)$.
A simple decomposition of $G$ takes as input a split $(U,W)$, and it outputs the two subgraphs $G_U = G[U \cup \{w\}]$ and $G_W = G[W \cup \{u\}]$ with $u \in N_G(W), \ w \in N_G(U)$ being chosen arbitrarily.
The vertices $u,w$ are termed {\em split marker vertices}.
A {\em split decomposition} of $G$ is obtained by applying recursively some sequence of simple decompositions ({\it e.g.}, see Fig.~\ref{fig:split-dec}).
We name {\em split components} the subgraphs in a given split decomposition of $G$.

\begin{figure}[h!]
\centering
\begin{subfigure}[b]{.46\textwidth}\centering
\includegraphics[width=.45\textwidth]{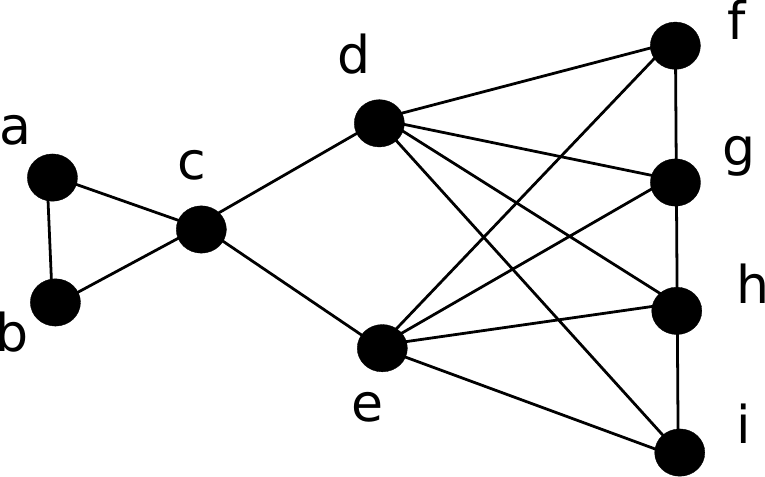}
\end{subfigure}\hfill
\begin{subfigure}[b]{.52\textwidth}\centering
\includegraphics[width=.75\textwidth]{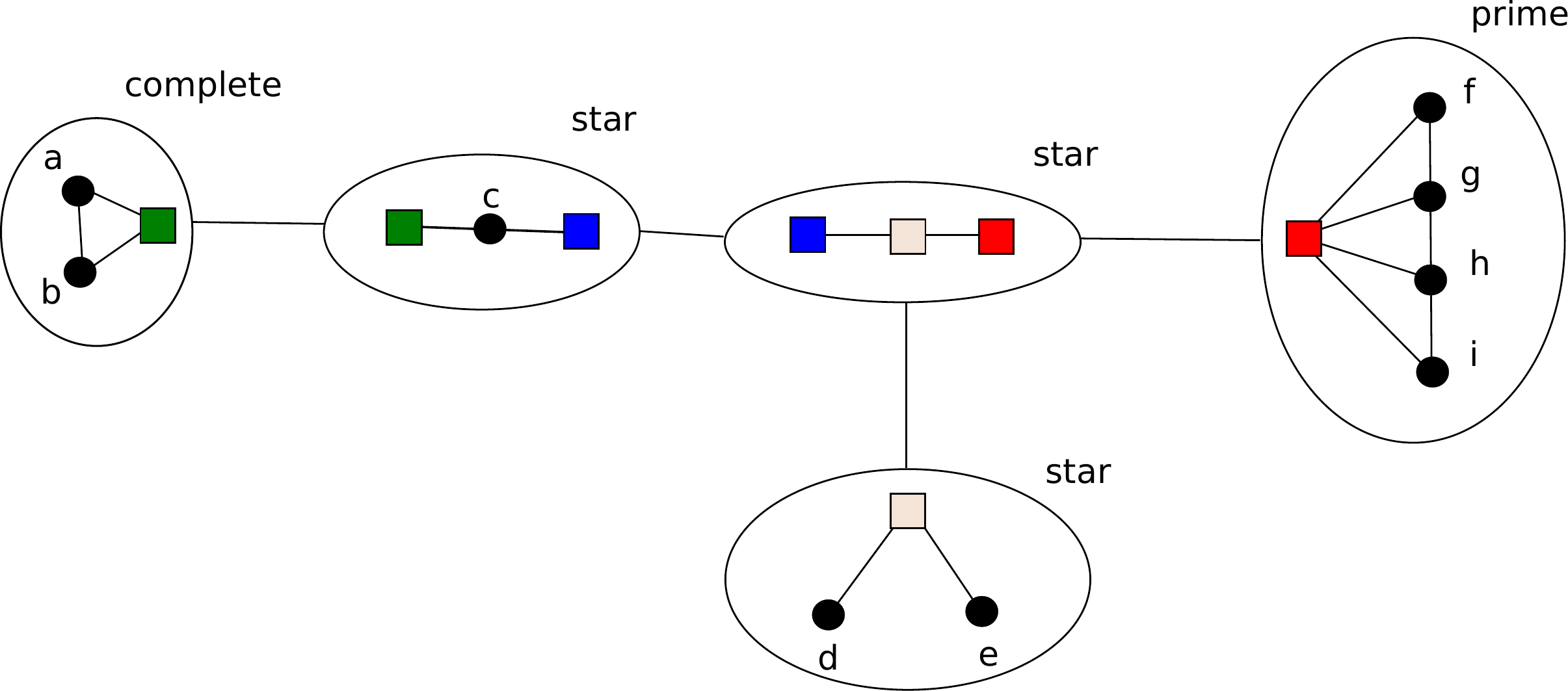}
\end{subfigure}
\caption{A graph and its split decomposition.}
\label{fig:split-dec}
\end{figure}

It is often desirable to apply simple decompositions until all the subgraphs obtained cannot be further decomposed.
In the literature there are two cases of ``indecomposable'' graphs.
Degenerate graphs are such that every bipartition of their vertex-set is a split.
They are exactly the complete graphs and the stars~\cite{Cun82}.
A graph is prime for split decomposition if it has no split.
We can define the following two types of split decomposition:
\begin{itemize}
\item{\bf Canonical split decomposition.}
Every graph has a canonical split decomposition where all the subgraphs obtained are either degenerate or prime and the number of subgraphs is minimized.
Furthermore, the canonical split decomposition of a given graph can be computed in linear-time~\cite{CDR12}.

\item{\bf Minimal split decomposition.}
A split-decomposition is {\em minimal} if all the subgraphs obtained are prime.
A minimal split-decomposition can be computed from the canonical split-decomposition in linear-time~\cite{Cun82}, that consists in performing simple decompositions on every split component that induces either a star or a complete graph of order at least four, until all such components have been replaced by a collection of triangles and paths of length three.
Doing so, we avoid handling with the particular cases of stars and complete graphs in our algorithms.
The set of prime graphs in any minimal split decomposition is unique up to isomorphism~\cite{Cun82}.
\end{itemize}

For instance, the split decomposition of Fig.~\ref{fig:split-dec} is both minimal and canonical.

\begin{definition}\label{def:sw}
The {\em split-width}
of $G$, denoted by $sw(G)$, is the minimum $k \geq 2$ such that any prime subgraph in the canonical split decomposition of $G$ has order at most $k$.
\end{definition}

We refer to~\cite{Rao08b} for some algorithmic applications of split decomposition.
In particular, graphs with split-width at most two are exactly the distance-hereditary graphs, {\it a.k.a} the graphs whose all connected induced subgraphs are distance-preserving~\cite{BaM86,GiP12}.
Distance-hereditary graphs contain many interesting subclasses of their own such as {\em cographs} ({\it a.k.a.}, $P_4$-free graphs) and $3$-leaf powers, see~\cite{GiP12}.
Furthermore, we can observe that every component in a minimal split decomposition of $G$ has order at most $\max\{3,sw(G)\}$.

\paragraph{\sc Split decomposition tree.}
A split decomposition tree of $G$ is a tree $T$ where the nodes are in bijective correspondance with the subgraphs of a given split decomposition of $G$, and the edges of $T$ are in bijective correspondance with the simple decompositions used for their computation.
More precisely:
\begin{itemize}
\item If the considered split decomposition is reduced to $G$ then $T$ is reduced to a single node;
\item Otherwise, let $(U,W)$ be a split of $G$ and let $G_U = (U \cup \{w\}, E_U), \ G_W = (W \cup \{u\}, E_W)$ be the corresponding subgraphs of $G$.
We construct the split decomposition trees $T_U, T_W$ for $G_U$ and $G_W$, respectively.
Furthermore, the split marker vertices $u$ and $w$ are contained in a unique split component of $G_W$ and $G_U$, respectively.
We obtain $T$ from $T_U$ and $T_W$ by adding an edge between the two nodes that correspond to these subgraphs.
\end{itemize}

The split decomposition tree of the canonical split decomposition, resp. of a minimal split decomposition, can be constructed in linear-time~\cite{Rao08b}.

\subsection*{Matching problems}

A matching in a graph is a set of edges with pairwise disjoint end vertices.
We consider the problem of computing a matching of maximum cardinality.

\begin{center}
	\fbox{
		\begin{minipage}{.95\linewidth}
			\begin{problem}[\textsc{Maximum Matching}]\
				\label{prob:matching} 
					\begin{description}
					\item[Input:] A graph $G=(V,E)$.
					\item[Output:] A matching of $G$ with maximum cardinality.
				\end{description}
			\end{problem}     
		\end{minipage}
	}
\end{center}

The {\sc Maximum Matching} problem can be solved in ${\cal O}(m\sqrt{n})$-time (see~\cite{Gab17} for a simplified presentation).
We do not use this result directly in our paper.
However, we do use in our analysis the notion of {\em augmenting paths}, that is a cornerstone of most matching algorithms.
Formally, let $G=(V,E)$ be a graph and $F \subseteq E$ be a matching of $G$.
A vertex is termed matched if it is incident to an edge of $F$, and exposed otherwise.
An $F$-augmenting path is a path where the two ends are exposed, all edges $\{v_{2i},v_{2i+1}\}$ are in $F$ and all edges $\{v_{2j-1}, v_{2j}\}$ are not in $F$.
We can observe that, given an $F$-augmenting path $P = (v_1,v_2, \ldots, v_{2\ell})$, the matching $E(P)\Delta F$ (obtained by replacing the edges $\{v_{2i},v_{2i+1}\}$ with the edges $\{v_{2j-1}, v_{2j}\}$) has larger cardinality than $F$.

\begin{lemma}[Berge,~\cite{Ber57}]\label{lem:berge}
A matching $F$ in $G=(V,E)$ is maximum if and only if there is no $F$-augmenting path.
\end{lemma}

It is folklore that the proof of Berge's lemma also implies the existence of many vertex-disjoint augmenting paths for small matchings.
We will use the following result in our analysis:

\begin{lemma}[Hopcroft-Karp,~\cite{HoK73}]\label{lem:hopcroft-karp}
Let $F_1,F_2$ be matchings in $G=(V,E)$.
If $|F_1| = r, \ |F_2| = s$ and $s > r$, then there exist at least $s-r$ vertex-disjoint $F_1$-augmenting paths.
\end{lemma}

\paragraph{\sc $b$-Matching.}
More generally given a graph $G=(V,E)$, let $b : V \to \mathbb{N}$ assign a nonnegative integer capacity $b_v$ for every vertex $v \in V$.
A $b$-matching is an assignment of nonnegative integer edge-weights $(x_e)_{e \in E}$ such that, for every $v \in V$, we have $\sum_{e \in E_v} x_{e} \leq b_v$.
We define the $x$-degree of vertex $v$ as $deg_x(v) = \sum_{e \in E_v} x_{e}$.
Furthermore, the cardinality of a $b$-matching is defined as $||x||_1 = \sum_{e \in E} x_e$.
We will consider the following graph problem:

\begin{center}
	\fbox{
		\begin{minipage}{.95\linewidth}
			\begin{problem}[$b$-\textsc{Matching}]\
				\label{prob:b-matching} 
					\begin{description}
					\item[Input:] A graph $G=(V,E)$; an assignment function $b : V \to \mathbb{N}$.
					\item[Output:] A $b$-matching of $G$ with maximum cardinality.
				\end{description}
			\end{problem}     
		\end{minipage}
	}
\end{center}

For technical reasons, we will also need to consider at some point the following variant of $b$-\textsc{Matching}.
Let $c : E \to \mathbb{N}$ assign a cost to every edge.
The cost of a given $b$-matching $x$ is defined as $c \cdot x = \sum_{e \in E} c_ex_e$.

\begin{center}
	\fbox{
		\begin{minipage}{.95\linewidth}
			\begin{problem}[{\sc Maximum-Cost} $b$-\textsc{Matching}]\
				\label{prob:maxcost-b-matching} 
					\begin{description}
					\item[Input:] A graph $G=(V,E)$; assignment functions $b : V \to \mathbb{N}$ and $c : E \to \mathbb{N}$.
					\item[Output:] A maximum-cardinality $b$-matching of $G$ where the cost is maximized.
				\end{description}
			\end{problem}     
		\end{minipage}
	}
\end{center}

\begin{lemma}[~\cite{Gab83,Gab16}]\label{lem:b-matching}
For every $G=(V,E)$ and $b : V \to \mathbb{N}$, $c : E \to \mathbb{N}$, we can solve {\sc Maximum-Cost} $b$-\textsc{Matching} in ${\cal O}(nm\log^2{n})$-time.

In particular, we can solve $b$-\textsc{Matching} in ${\cal O}(nm\log^2{n})$-time.
\end{lemma}

There is a nonefficient (quasi polynomial) reduction from $b$-\textsc{Matching} to {\sc Maximum Matching} that we will use in our analysis ({\it e.g.}, see~\cite{Tut54}).
More precisely, let $G,b$ be any instance of $b$-\textsc{Matching}.
The ``expanded graph'' $G_b$ is obtained from $G$ and $b$ as follows.
For every $v \in V$, we add the nonadjacent vertices $v_1, v_2, \ldots, v_{b_v}$ in $G_b$.
Then, for every $\{u,v\} \in E$, we add the edges $\{u_i,v_j\}$ in $G_b$, for every $1 \leq i \leq b_u$ and for every $1 \leq j \leq b_v$.
It is easy to transform any $b$-matching of $G$ into an ordinary matching of $G_b$, and vice-versa.

\section{Changing the capacity of one vertex}\label{sec:var-1-vertex}

We first consider an auxiliary problem on $b$-matching that can be of independent interest.
Let $G=(V,E)$ be a graph, let $w \in V$ and let $b : V \setminus w \to \mathbb{N}$ be a partial assignment.
We denote $\mu(t)$ the maximum cardinality of a $b$-matching of $G$ provided we set to $t$ the capacity of vertex $w$.
Clearly, $\mu$ is nondecreasing in $t$. 
Our main result in this section is that the function $\mu$ is essentially piecewise linear (Proposition~\ref{prop:monotonic}).
We start introducing some useful lemmata.

\begin{lemma}\label{lem:gap}
$\mu(t+1) - \mu(t) \leq 1$.
\end{lemma}

\begin{proof}
Consider the two expanded graphs $G_{b,t}$ and $G_{b,t+1}$ that are obtained after setting the capacity of $w$ to, respectively, $t$ and $t+1$.
Let $F_t$ be any maximum matching of $G_{b,t}$.
By the hypothesis, $|F_t| = \mu(t)$.
Furthermore, since $G_{b,t+1}$ can be obtained from $G_{b,t}$ by adding a new vertex $w_{t+1}$ (that is a false twin of $w_1,w_2,\ldots,w_t$), $F_t$ is also a matching of $G_{b,t+1}$.
Suppose that $F_t$ is not maximum in $G_{b,t+1}$, {\it i.e.}, $\mu(t+1) = \mu(t) + r$, for some $r \geq 1$.
By Hopcroft-Karp lemma (Lemma~\ref{lem:hopcroft-karp}) there are $r$ $F_t$-augmenting paths in $G_{b,t+1}$ that are vertex-disjoint.
Furthermore, since $F_t$ is maximum in $G_{b,t}$, every such path must contain $w_{t+1}$, and so, $r=1$.
\end{proof}

\begin{lemma}\label{lem:after-2}
If $\mu( t + 2 ) = \mu( t )$ then we have $\mu(t + i) = \mu(t)$ for every $i \geq 0$.
\end{lemma}

\begin{proof}
By contradiction, let $i_0$ be the least integer $i \geq 3$ such that $\mu( t + i) > \mu(t)$.
By Lemma~\ref{lem:gap} we have $\mu(t+i_0) \leq \mu(t+i_0-1) + 1 = \mu(t) + 1$, therefore $\mu(t+i_0) = \mu(t)+1$.
Consider the two expanded graphs $G_{b,t}$ and $G_{b,t+i_0}$ that are obtained after setting the capacity of $w$ to, respectively, $t$ and $t+i_0$.
Let $F_t$ be any maximum matching of $G_{b,t}$.
By the hypothesis, $|F_t| = \mu(t)$.
Furthermore, since $G_{b,t+i_0}$ can be obtained from $G_{b,t}$ by adding the new vertices $w_{t+1}, \ldots, w_{t+i_0}$ (that are false twins of $w_1,w_2,\ldots,w_t$), $F_t$ is also a matching of $G_{b,t+i_0}$.
However $F_t$ is not maximum in $G_{b,t+i_0}$ since we have $\mu(t+i_0) = \mu(t) + 1$.
By Berge lemma (Lemma~\ref{lem:berge}) there exists an $F_t$-augmenting path $P$ in $G_{b,t+i_0}$.
Furthermore, since $F_t$ is maximum in $G_{b,t}$, $P$ must contain a vertex amongst $w_{t+1}, \ldots, w_{t+i_0}$.
Note that the latter vertices are all exposed, and so, an $F_t$-augmenting path can only contain at most two of them.
Furthermore, since $w_{t+1}, \ldots, w_{t+i_0}$ are pairwise false twins, we can assume w.l.o.g. that $P$ has $w_{t+1}$ as an end.
In the same way, in case $P$ has its two ends amongst $w_{t+1}, \ldots, w_{t+i_0}$ then we can assume w.l.o.g. that the two ends of $P$ are exactly $w_{t+1},w_{t+2}$.
It implies either $\mu(t) < \mu(t+1) \leq \mu(t+2)$ or $\mu(t) = \mu(t+1) < \mu(t+2)$.
In both cases, $\mu(t+2) > \mu(t)$, that is a contradiction.
\end{proof}

\begin{lemma}\label{lem:after-1}
If $\mu( t + 1 ) = \mu( t  )$ then we have $\mu( t + 3 ) = \mu( t + 2 )$.
\end{lemma}

\begin{proof}
If $\mu( t + 2 ) = \mu( t  )$ then the result follows from Lemma~\ref{lem:after-2} directly.
Thus, we assume from now on that $\mu( t + 2 ) > \mu( t  )$.
By Lemma~\ref{lem:gap} we have $\mu(t+2) \leq \mu(t+1) + 1 = \mu(t) + 1$, therefore $\mu(t+2) = \mu(t) + 1$.
Suppose by contradiction $\mu( t + 3 ) > \mu( t + 2 )$.
Again by Lemma~\ref{lem:gap} we get $\mu(t+3) = \mu(t+2) + 1$, therefore $\mu(t+3) = \mu(t) + 2$.
Consider the two expanded graphs $G_{b,t}$ and $G_{b,t+3}$ that are obtained after setting the capacity of $w$ to, respectively, $t$ and $t+3$.
Let $F_t$ be any maximum matching of $G_{b,t}$.
By the hypothesis, $|F_t| = \mu(t)$.
Furthermore, since $G_{b,t+3}$ can be obtained from $G_{b,t}$ by adding the three new vertices $w_{t+1}, w_{t+2}, w_{t+3}$ (that are false twins of $w_1,w_2,\ldots,w_t$), $F_t$ is also a matching of $G_{b,t+3}$.
However $F_t$ is not maximum in $G_{b,t+3}$ since we have $\mu(t+3) = \mu(t) + 2$.
By Hopcroft-Karp lemma (Lemma~\ref{lem:hopcroft-karp}) there exist $2$ $F_t$-augmenting paths in $G_{b,t+3}$ that are vertex-disjoint.
Furthermore, since $F_t$ is maximum in $G_{b,t}$, every such path must contain a vertex amongst $w_{t+1}, w_{t+2}, w_{t+3}$.
Note that the latter vertices are all exposed, and so, an $F_t$-augmenting path can only contain at most two of them.
In particular, there is at least one of these two $F_t$-augmenting paths, denoted by $P$, that only contains a single vertex amongst $w_{t+1}, w_{t+2}, w_{t+3}$.
W.l.o.g., since $w_{t+1}, w_{t+2}, w_{t+3}$ are pairwise false twins, we can assume that $w_{t+1}$ is an end of $P$.
However, it implies that $P$ is also an $F_t$-augmenting path in $G_{b,t+1}$, and so, that $\mu(t+1) > \mu(t)$, that is a contradiction.
\end{proof}

We are now ready to prove the main result of this section:

\begin{proposition}\label{prop:monotonic}
There exist integers $c_1,c_2$ such that:
$$\mu(t) = \begin{cases}
\mu(0) + t \ \mbox{if} \ t \leq c_1 \\
\mu(c_1) + \left\lfloor (t-c_1)/2 \right\rfloor = \mu(0) + c_1 +  \left\lfloor (t-c_1)/2 \right\rfloor \ \mbox{if} \ c_1 < t \leq c_1 + 2c_2 \\
\mu(c_1+2c_2) = \mu(0) + c_1 + c_2 \ \mbox{otherwise}.
\end{cases}$$
Furthermore, the triple $(\mu(0),c_1,c_2)$ can be computed in ${\cal O}(nm\log^2{n}\log{||b||_1})$-time.
\end{proposition}

\begin{figure}[h!]
\centering
\includegraphics[width=.8\textwidth]{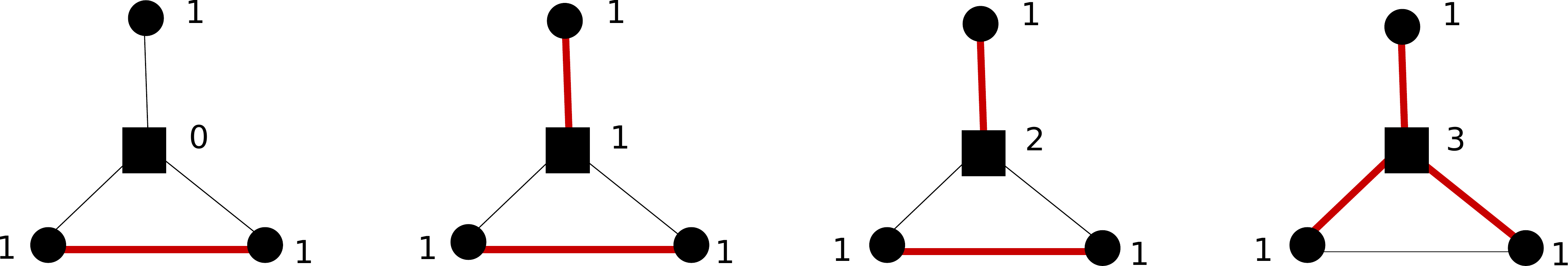}
\caption{An example with $(\mu(0),c_1,c_2) = (1,1,1)$. Vertices are labeled with their capacity. Thin and bold edges have respective weights $0$ and $1$.}
\label{fig:sreduction}
\end{figure}

\begin{proof}
Let $c_1$ be the maximum integer $t$ such that $\mu(t) = \mu(0) + t$.
This value is well-defined since, trivially, the function $\mu$ must stay constant whenever $t \geq \sum_{v \in N_G(w)} b_v$.
Furthermore, by Lemma~\ref{lem:gap} we have $\mu(t) = \mu(0) + t$ for every $0 \leq t \leq c_1$.
Then, let $t_{\max}$ be the least integer $t$ such that, for every $i \geq 0$ we have $\mu(t_{\max}+i) = \mu(t_{\max})$.
Again, this value is well-defined since we have the trivial upper-bound $t_{\max} \leq \sum_{v \in N_G(w)} b_v$.
Furthermore, since $\mu$ is strictly increasing on $[|0;c_1|]$, we have $t_{\max} \geq c_1$.
Let $c_2' = t_{\max} - c_1$.
We claim that $c_2' = 2c_2$ is even.
For that, we need to observe that $\mu(c_1) = \mu(c_1+1)$ by maximality of $c_1$.
By Lemma~\ref{lem:after-1} we get $\mu(c_1+2i) = \mu(c_1+2i+1)$ for every $i \geq 0$, thereby proving the claim.
Moreover, for every $0 \leq i < c_2$ we have by Lemma~\ref{lem:after-2} $\mu(c_1 + 2i) < \mu(c_1 + 2(i+1))$ (since otherwise $t_{\max} \leq c_1 + 2i$).
By Lemma~\ref{lem:after-1} we have $\mu(c_1+2i) = \mu(c_1+2i+1)$.
Finally, by Lemma~\ref{lem:gap} we get $\mu(c_1+2(i+1)) \leq \mu(c_1+2i+1) +1 = \mu(c_1+2i)+1$, therefore $\mu(c_1+2(i+1)) = \mu(c_1+2i) + 1$.
Altogether combined, it implies that $\mu(c_1+2i) = \mu(c_1+2i+1) = \mu(c_1) + i$ for every $0 \leq i \leq c_2$, that proves the first part of our result.

We can compute $\mu(0)$ with any $b$-{\sc Matching} algorithm after we set the capacity of $w$ to $0$.
Furthermore, the value of $c_1$ can be computed within ${\cal O}(\log c_1)$ calls to a $b$-{\sc Matching} algorithm, as follows.
Starting from $c_1' = 1$, at every step we multiply the current value of $c_1'$ by $2$ until we reach a value $c_1' > c_1$ such that $\mu(c_1') < \mu(0) + c_1'$.
Then, we perform a dichotomic search between $0$ and $c_1'$ in order to find the largest value $c_1$ such that $\mu(c_1) = \mu(0) + c_1$.
Once $c_1$ is known, we can use a similar approach in order to compute $c_2$.
Overall, since $c_1 + 2c_2 = t_{\max} \leq \sum_{v \in N_G(w)} b_v = {\cal O}(||b||_1)$, we are left with ${\cal O}(\log{||b||_1})$ calls to any $b$-{\sc Matching} algorithm.
Therefore, by Lemma~\ref{lem:b-matching}, we can compute the triple $(\mu(0),c_1,c_2)$ in ${\cal O}(nm\log^2{n}\log{||b||_1})$-time.
\end{proof}

\section{The algorithm}\label{sec:algo}

We present in this section a quasi linear-time algorithm for computing a maximum-cardinality $b$-matching on any bounded split-width graph (Theorem~\ref{thm:main-result}).
Given a graph $G$, our algorithm takes as input the split decomposition tree $T$ of any minimal split decomposition of $G$.
We root $T$ in an arbitrary component $C_1$.
Then, starting from the leaves, we compute by dynamic programming on $T$ the {\it cardinality} of an optimal solution.
This first part of the algorithm is involved, and it uses the results of Section~\ref{sec:var-1-vertex}.
It is based on a new reduction rule that we introduce in Definition~\ref{def:reduction-rule}.
Finally, starting from the root component $C_1$, we compute a maximum-cardinality $b$-matching of $G,b$ by reverse dynamic programming on $T$.
This second part of the algorithm is simpler than the first one, but we need to carefully upper-bound its time complexity.
In particular, we also need to ensure that some additional property holds for the $b$-matchings we compute at every component.

\subsection*{Reduction rule}

\begin{definition}\label{def:reduction-rule}
Let $G = (V,E),b$ be any instance of $b$-{\sc Matching}.
For any split $(U,W)$ of $G$ let $C = N_G(W) \subseteq U, \ D = N_G(U) \subseteq W$.
Furthermore, let $G_U = (U \cup \{w\}, E_U), \ G_W = (W \cup \{u\}, E_W)$ be the corresponding subgraphs of $G$.
We define the pairs $G_U,b^U$ and $H_W,b^W$ as follows:

\begin{figure}[h!]
\centering
\includegraphics[width=.4\textwidth]{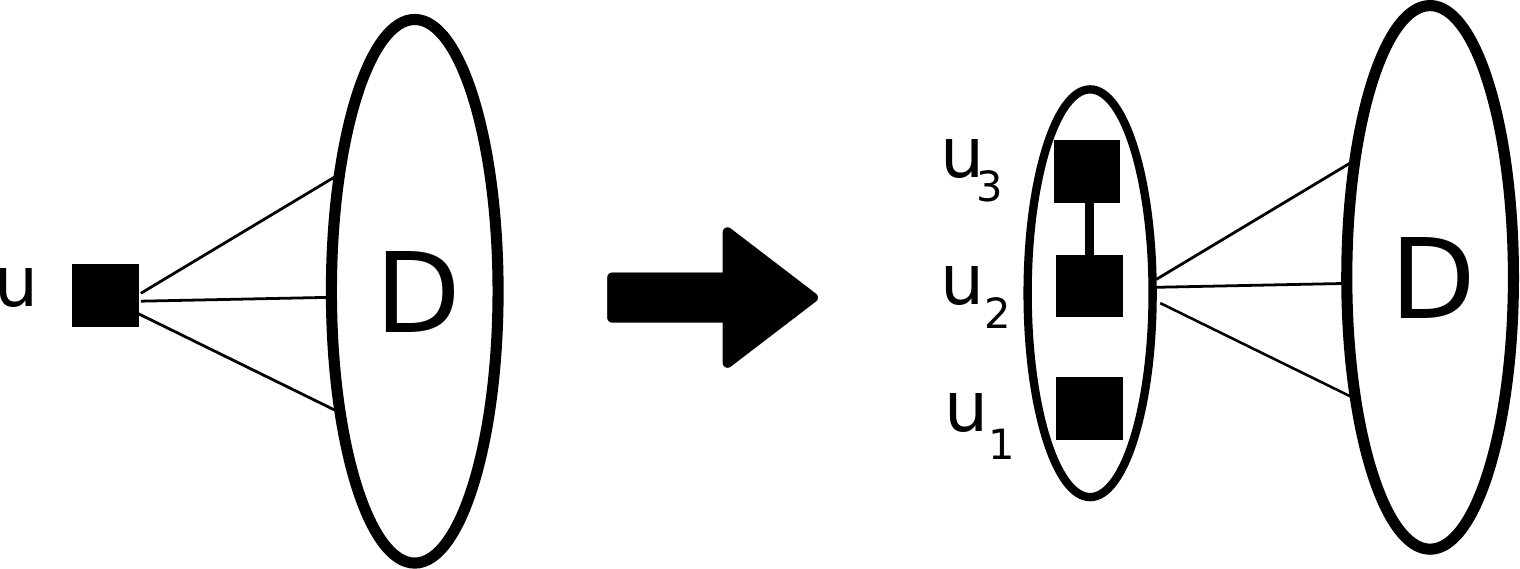}
\caption{The reduction of Definition~\ref{def:reduction-rule}.}
\label{fig:sreduction}
\end{figure}

\begin{itemize}
\item For every $v \in U$ we set $b^U_v = b_v$; the capacity of the split marker vertex $w$ is left unspecified.
Let the triple $(\mu^U(0),c_1^U,c_2^U)$ be as defined in Proposition~\ref{prop:monotonic} w.r.t. $G_U,b^U$ and $w$. 
\item The {\em auxiliary graph} $H_W$ is obtained from $G_W$ by replacing the split marker vertex $u$ by a module $M_u = \{u_1,u_2,u_3\}, \ N_{H_W}(M_u) = N_{G_W}(u) = D$; we also add an edge between $u_2,u_3$ (cf. Fig.~\ref{fig:sreduction}). 
For every $v \in W$ we set $b^W_v = b_v$; we set $b^W_{u_1} = c_1^U, \ b^W_{u_2} = b^W_{u_3} = c_2^U$.
\end{itemize}
\end{definition}

We will show throughout this section that our gadget somewhat encodes all the partial solutions for side $U$.
For that, we start proving the following lemma:

\begin{lemma}\label{lem:reduction-rule-1}
Let $x$ be a $b$-matching for $G,b$.
There exists a $b$-matching $x^W$ for $H_W,b^W$ such that $||x^W||_1 \geq ||x||_1 + c_2^U - \mu^U(0)$.
\end{lemma}

\begin{proof}
See Fig.~\ref{fig:reduction-1} for an illustration.
Let us define $b^U_w = \sum_{e \in C \times D} x_e$.
As an intermediate step, we can construct a $b$-matching $x^U$ of the pair $G_U, b^U$ as follows.
First we keep all the edge-weights $x_e = x^U_e$, for every $e \in E(U)$.
Then, for every $v \in C$ we set $x^U_{\{v,w\}} = \sum_{v' \in D} x_{\{v,v'\}}$.
We deduce from this transformation that: $$||x^U||_1 = \sum\limits_{e \in E(U) \cup (C \times D)} x_e \leq \mu^U(b_w^U).$$
In particular, let $y^U$ be a $b$-matching of $G_U, b^U$ of optimal cardinality $\mu^U(b_w^U)$ and such that $d = deg_{y^U}(w) = \sum_{e \in E_w(G_U)} y^U_e$ is minimized. 
By Proposition~\ref{prop:monotonic}, we have $d \leq c_1^U + 2c_2^U$.
We obtain a $b$-matching $x^W$ for the pair $H_W,b^W$ as follows:

\begin{figure}[h!]
\centering
\includegraphics[width=.6\textwidth]{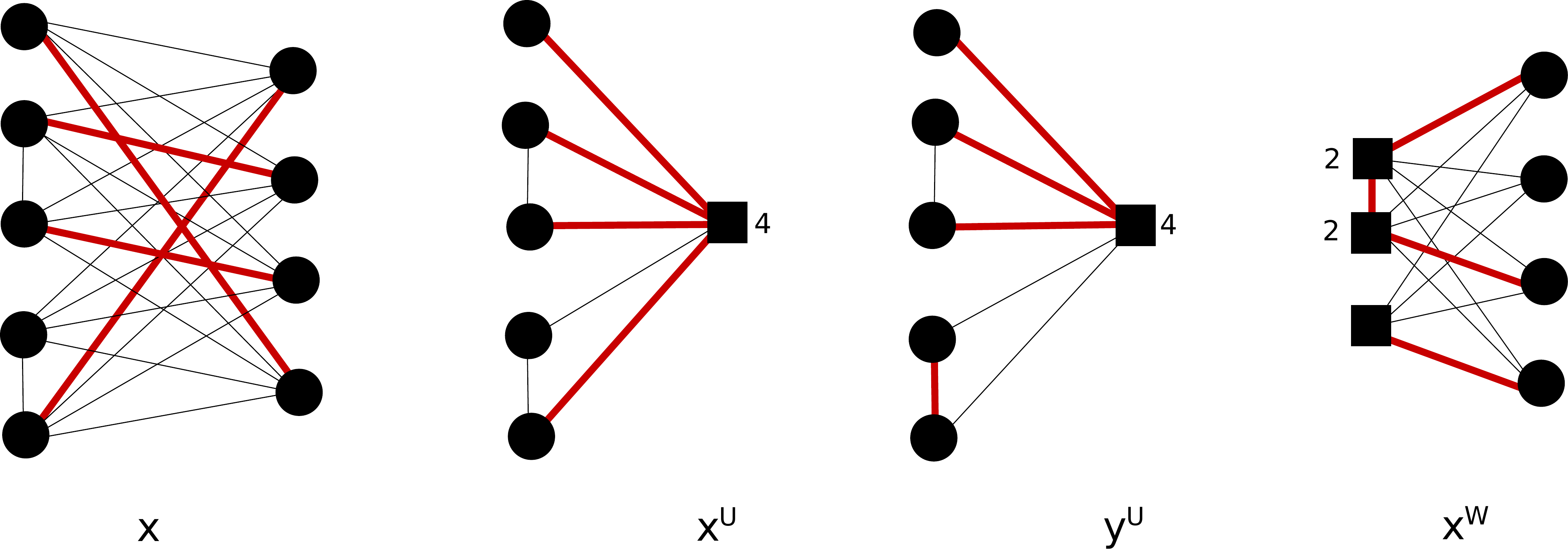}
\caption{The construction of $x^W$. Vertices with capacity greater than $1$ are labeled with their capacity. Thin and bold edges have respective weights $0$ and $1$.}
\label{fig:reduction-1}
\end{figure}

\begin{itemize}
\item We keep all the edge-weights $x_e = x^W_e$, for every $e \in E(W)$.
Doing so, we get an initial $b$-matching of cardinality $||x||_1 - ||x^U||_1 \geq ||x||_1 - \mu^U(b_w^U)$;
\item Then, in order to define $x_e^W$, for every edge $e$ that is incident to $u_1$, we make a simple reduction to a flow problem.
Namely, consider a star with leaves $D$ and central node the split marker vertex $u$.
The star is node-capacitated, with the capacity being set to: $\min\{d,c_1^U\}$ for $u$; and $\sum_{v \in C} x_{\{v,v'\}}$ for every $v' \in D$.
By construction, the capacities on side $D$ sum to $b_w^U \geq d$.
So, we can send exactly $\min\{d,c_1^U\}$ units of flow from $D$ to $u$.
Furthermore, we define $x^W_{\{u_1,v'\}}$, for every $v' \in D$, as the value of the flow passing through the arc $(v',u)$.
Doing so, we get a partial solution of cardinality $\geq ||x||_1 - \mu^U(b_w^U) + \min\{d,c_1^U\}$.
\item We end up defining the weights of the edges that are incident to $u_2$ or $u_3$.
There are two cases:
\begin{itemize}
\item{\bf Case $d \leq c_1^U$.}
By Proposition~\ref{prop:monotonic}, $\mu^U(b_w^U) - d = \mu^U(0)$.
We complete the $b$-matching by setting $x^W_{\{u_2,u_3\}} = c_2^U$.
\item{\bf Case $d > c_1^U$.}
By Proposition~\ref{prop:monotonic} and the minimality of $d$, we have that $d - c_1^U = 2d'$ is even, $d' \leq c_2^U$ and $\mu^U(b_w^U) - c_1^U = \mu^U(0) + d'$.
We make two other reductions to a flow problem, on the same star as before but with different node capacities.
More precisely, we set the capacity of $u$ to $d'$, while for every $v' \in D$, we decrease its original capacity by $x^W_{\{u_1,v'\}}$.
By construction, the capacities on side $D$ now sum to $b_w^U - c_1^U \geq d - c_1^U \geq 2d'$.
So, we can send exactly $d'$ units of flow from $D$ to $u$.
Furthermore, we define $x^W_{\{u_2,v'\}}$, for every $v' \in D$, as the value of the flow passing through the arc $(v',u)$.

Then, we again decrease the node capacity for every $v' \in D$, by exactly $x^W_{\{u_2,v'\}}$.
We send $d'$ more units of flow from $D$ to $u$.
For every $v' \in D$, we define $x^W_{\{u_3,v'\}}$ as the additional amount of flow being sent through the arc $(v',u)$.
Note that, doing so, we get that $\sum_{v' \in D} x^W_{\{u_2,v'\}} = \sum_{v' \in D} x^W_{\{u_3,v'\}} = d'$.
Finally, we set $x_{\{u_2,u_3\}}^W = c_2^U - d'$.
In total, the cardinality of the $b$-matching has so increased by $2d' + (c_2^U - d') = c_2^U + d'$.
\end{itemize}
Therefore in both cases, the resulting $b$-matching has cardinality at least $||x||_1 + c_2^U - \mu^U(0)$.
\end{itemize}
\end{proof}

In fact, the converse of Lemma~\ref{lem:reduction-rule-1} also holds: if $x^W$ is a $b$-matching for $H_W,b^W$ then there exists a $b$-matching $x$ for the pair $G,b$ such that $||x||_1 \geq ||x^W||_1 - c_2^U + \mu^U(0)$ (Proposition~\ref{prop:reduction-rule}).
We postpone the proof of the converse inequality since, in order to prove it, we first need to prove intermediate lemmata that will be also used in the proof of Theorem~\ref{thm:main-result}.

\subsection*{$b$-matchings with additional properties}

We consider an intermediate modification problem on the $b$-matchings of some ``auxiliary graphs'' that we define next.
Let $C_i$ be a split component in a given split decomposition of $G$.
The subgraph $C_i$ is obtained from a sequence of simple decompositions.
For a given subsequence of the above simple decompositions (to be defined later) we apply the reduction rule of Definition~\ref{def:reduction-rule}.
Doing so, we obtain a pair $H_i,b^i$ with $H_i$ being a supergraph of $C_i$ obtained by replacing some split marker vertices $u_{i_t}, \ 1 \leq t \leq l$, by the modules $M_{i_t} = \{u_{i_t}^1,u_{i_t}^2,u_{i_t}^3\}$.
We recall that $u_{i_t}^2,u_{i_t}^3$ are adjacent and they have the same capacity.

\smallskip
We seek for a maximum-cardinality $b$-matching $x^i$ for the pair $H_i,b^i$ such that the following properties hold for every $1 \leq t \leq l$:
\begin{itemize}
\item ({\bf symmetry}) $deg_{x^i}(u_{i_t}^2) = deg_{x^i}(u_{i_t}^3)$. 
\item ({\bf saturation}) if $deg_{x^i}(u_{i_t}^1) < c_{i_t}^1$ then, $deg_{x^i}(u_{i_t}^2) = x^i_{\{u_{i_t}^2,u_{i_t}^3\}}$.
\end{itemize} 

We prove next that for every fixed $t$, any $x^i$ can be processed in ${\cal O}(|E_{u_{i_t}}(C_i)|)$-time so that both the saturation property and the symmetry property hold for $M_{i_t}$.
However, ensuring that these two above properties hold {\em simultaneously} for every $t$ happens to be trickier.

\begin{lemma}\label{lem:simultaneous}
In ${\cal O}(|V(H_i)| \cdot |E(H_i)| \cdot \log^2{|V(H_i)|})$-time, we can compute a maximum-cardinality $b$-matching $x^i$ for the pair $H_i,b^i$ such that both the saturation property and the symmetry property hold for every $M_{i_t}, \ 1 \leq t \leq l$.
\end{lemma}

\begin{proof}
Let $x^i$ be some initial maximum-cardinality $b$-matching (to be defined later).
While there exists a $t$ such that the saturation property or the symmetry property does not hold for $M_{i_t}$, we keep repeating the following rules until none of them can be applied:
\begin{itemize}

\item{\bf Rule 1}. Suppose $deg_{x^i}(u_{i_t}^1) < c_{i_t}^1$ and there exists $v' \in N_{H_i}(M_{i_t})$ such that $x^i_{\{u_{i_t}^2,v'\}} > 0$ (resp., $x^i_{\{u_{i_t}^3,v'\}} > 0$). Then, we increase $x^i_{\{u_{i_t}^1,v'\}}$ as much as we can, that is by exactly $\min\{c_{i_t}^1 - deg_{x^i}(u_{i_t}^1), x^i_{\{u_{i_t}^2,v'\}}\}$ (resp., $\min\{c_{i_t}^1 - deg_{x^i}(u_{i_t}^1), x^i_{\{u_{i_t}^3,v'\}}\}$), and we decrease $x^i_{\{u_{i_t}^2,v'\}}$ (resp., $x^i_{\{u_{i_t}^3,v'\}}$) by exactly the same amount.
By repeating this step until it is no more possible to do so, we ensure that the saturation property holds for $M_{i_t}$.

\item{\bf Rule 2}. Suppose $deg_{x^i}(u_{i_t}^2) > deg_{x^i}(u_{i_t}^3) + 1$ (the case $deg_{x^i}(u_{i_t}^3) > deg_{x^i}(u_{i_t}^2) + 1$ is symmetrical to this one). Let $v' \in N_{H_i}(M_{i_t})$ such that $x^i_{\{u_{i_t}^2,v'\}} > x^i_{\{u_{i_t}^3,v'\}}$.
We increase $x^i_{\{u_{i_t}^3,v'\}}$ as much as we can, that is by $\min\left\{\left\lfloor \frac {deg_{x^i}(u_{i_t}^2) - deg_{x^i}(u_{i_t}^3)} 2 \right\rfloor,x^i_{\{u_{i_t}^3,v'\}}\right\}$, while we decrease $x^i_{\{u_{i_t}^3,v'\}}$ by exactly the same amount.
By repeating this step until it is no more possible to do so, we ensure that $|deg_{x^i}(u_{i_t}^2) - deg_{x^i}(u_{i_t}^3)| \leq 1$.

\item{\bf Rule 3}. Suppose $deg_{x^i}(u_{i_t}^2) = deg_{x^i}(u_{i_t}^3) + 1$ (the case $deg_{x^i}(u_{i_t}^3) = deg_{x^i}(u_{i_t}^2) + 1$ is symmetrical to this one). Let $v' \in N_{H_i}(M_{i_t})$ such that $x^i_{\{u_{i_t}^2,v'\}} > x^i_{\{u_{i_t}^3,v'\}}$. We decrease $x^i_{\{u_{i_t}^2,v'\}}$ by one unit; similarly, we increase $x^i_{\{u_{i_t}^2,u_{i_t}^3\}}$ by one unit.
Doing so, we ensure that the symmetry property holds for $M_{i_t}$.
\end{itemize}
Overall, we only need to scan ${\cal O}(1)$ times the set $N_{H_i}(M_{i_t})$, and so, we can ensure that both the saturation property and the symmetry property hold for $M_{i_t}$ in ${\cal O}(|N_{H_i}(M_{i_t})|)$-time.
However, doing so, we may break the saturation property or the symmetry property for some other $t' \neq t$ ({\it e.g.}, see Fig.~\ref{fig:bad-example}).
Therefore, if we start from an arbitrary $x^i$, the above procedure may take quasi polynomial time in order to converge.

\begin{figure}[h!]
\centering
\includegraphics[width=.7\textwidth]{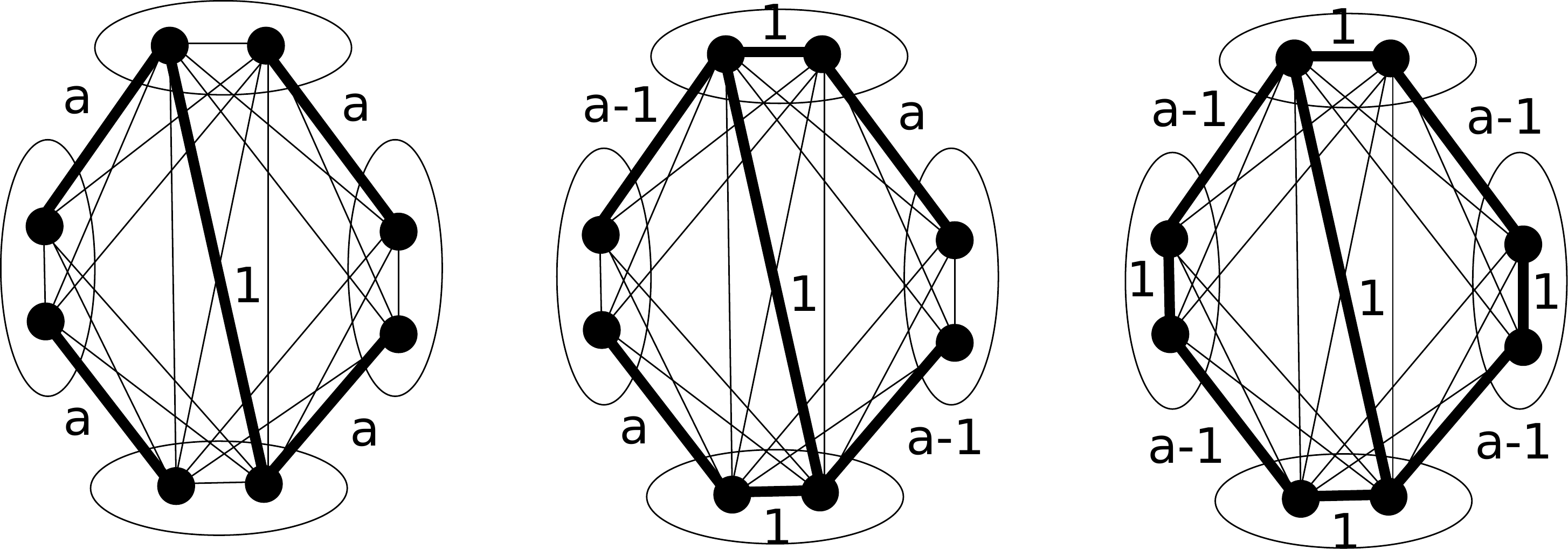}
\caption{An example where the naive processing stage requires ${\cal O}(||b||_1)$-time.}
\label{fig:bad-example}
\end{figure}

\medskip
In order to overcome this above issue, we assign costs on the edges. All the edges of $H_i$ have unit cost, except the edges $\{u_{i_t}^2,u_{i_t}^3\}$, for every $1 \leq t \leq l$, to which we assign cost $2$.
We compute a maximum-cardinality $b$-matching $x^i$ for the pair $H_i,b^{H_i}$ that is of {\em maximum cost}.
By Lemma~\ref{lem:b-matching}, this can be done in time ${\cal O}(|V(H_i)||E(H_i)|\log^2{|V(H_i)|})$.
Then, we apply the same procedure on $x^i$ as described above.
We claim that there is at most one loop for every $t$.
Indeed, let $1 \leq t \leq l$ be arbitrary.
We observe that Rules 1 and 2 do not change the cost of the $b$-matching.
Furthermore, for every $v' \in N_{H_i}(M_{i_t})$, Rules 1 and 2 do not change the $x^i$-degree of $v'$.
Hence, we can only break the saturation property or the symmetry property for some other $t' \neq t$ by applying Rule 3.
However, Rule 3 increases the cost of $x^i$, and so, since $x^i$ is supposed to be of maximum cost, this rule will never be applied.
Therefore, the claim is proved.
Overall, it implies that the postprocessing of $x^i$ takes time ${\cal O}(\sum_{t=1}^l |N_{H_i}(M_{i_t})|)$, that is in ${\cal O}(|E(H_i)|)$.
\end{proof}

\subsection*{Merging the partial solutions together}

Finally, before we can describe our main algorithm (Theorem~\ref{thm:main-result}) we need to consider the intermediate problem of merging two partial solutions (that have been obtained, respectively, for the two subgraphs outputted by a simple decomposition).
For that, we introduce the following data structure for storing a $b$-matching:

\begin{lemma}\label{lem:data-structure}
For every $G=(V,E)$ there exists a data structure that can be initialized in ${\cal O}(m)$-time, and such that the following properties hold:
\begin{itemize}
\item ({\bf Access}) An edge-weight assignment $(x_e)_{e \in E}$ is stored. Initially, all the edge-weights are set to $0$. For every $e \in E$, we can read and modify $x_e$ in constant-time.
\item ({\bf Split}) Furthermore, let $(U,W)$ be a split of $G$ and let $G_U = (U \cup \{w\}, E_U), \ G_W = (W \cup \{u\}, E_W)$ be the corresponding subgraphs of $G$. In ${\cal O}(1)$-time, we can modify the data structure so that it can store separate edge-weight assignments for $G_U$ and $G_W$ (initially set to $(x_e)_{e \in E_U}$ and $(x_e)_{e \in E_W}$).
\item ({\bf Merge}) Conversely, let $(U,W)$ be a split of $G$ and let $G_U = (U \cup \{w\}, E_U), \ G_W = (W \cup \{u\}, E_W)$ be the corresponding subgraphs of $G$. Suppose that two separate assignments $(x^U_e)_{e \in E_U}$ and $(x^W_e)_{e \in E_W}$ are stored in the data structure. In ${\cal O}(1)$-time, we can modify the data structure so that it stores the following edge-weight assignment for $G$:
$$x_e = \begin{cases}
x^U_e \ \mbox{if} \ e \in E_U \setminus  E_W \\
x^W_e \ \mbox{if} \ e \in E_W \setminus  E_U \\
\max\{x^U_e,x^W_e\} \ \mbox{if} \ e \in E_U \cap  E_W  \\
\mbox{undefined otherwise}.
\end{cases}$$
\end{itemize} 
\end{lemma}

\begin{proof}
Every edge $e \in E$ has a pointer to its corresponding weight $x_e$ (initially set to $0$).
Now, consider any split $(U,W)$ of $G$ and let $G_U = (U \cup \{w\}, E_U), \ G_W = (W \cup \{u\}, E_W)$ be the corresponding subgraphs of $G$.
Observe that $E_U,E_W$ intersect in exactly one edge $e_{U,W} = \{u,w\}$.
So, in order to perform a split of the data structure, it suffices to split the pointer of $e_{U,W}$ in two new pointers, that are initially pointing to two distinct copies of the value $x_e$.
Note that for every pointer newly introduced, we need to keep track of the corresponding split $(U,W)$ and of the corresponding side ({\it i.e.}, $U$ or $W$).
Conversely, in order to merge the respective data structures obtained for $G_U$ and $G_W$, it suffices to extract the values $x_{e_{U,W}}^U$ and $x_{e_{U,W}}^W$ (on which the two new pointers introduced after the split are pointing to) and to set the original value $x_{e_{U,W}}$ to $\max\{x_{e_{U,W}}^U,x_{e_{U,W}}^W\}$.
\end{proof}

Let $(U,W)$ be a split of $G$ and let $G_U = (U \cup \{w\}, E_U), \ G_W = (W \cup \{u\}, E_W)$ be the corresponding subgraphs of $G$.
Consider some partial solutions $x^U$ and $x^W$ obtained, respectively, for the pairs $G_U,b^U$ and $G_W,b^W$ (for some capacity functions $b^U,b^W$ to be defined later in the description of our main algorithm).
We reduce the merging stage to a flow problem on a complete bipartite subgraph (induced by the split) with node-capacities.
We detail this next.

\begin{lemma}\label{lem:potential-function}
Suppose that $b^U$ (resp., $b^W$) satisfies $b^U_v \leq b_v$ for every $v \in U$ (resp., $b^W_v \leq b_v$ for every $v \in W$).
Let $x^U,x^W$ a $b$-matching for, respectively, the pairs $G_U,b^U$ and $G_W,b^W$ such that $deg_{x^U}(w) = deg_{x^W}(u) = d$.

\smallskip
Furthermore, for any graph $H$ let $\varphi(H) = |E(H)| + 4 \cdot (sc(H) - 1)$, with $sc(H)$ being the number of split components in any minimal split decomposition of $H$\footnote{We recall that the set of prime graphs in any minimal split decomposition is unique up to isomorphism~\cite{Rao08b}.}.

\smallskip
Then, in at most ${\cal O}(\varphi(G) - \varphi(G_U) - \varphi(G_W))$-time, we can obtain a valid $b$-matching $x$ for the pair $G,b$ such that $||x||_1 = ||x^U||_1+||x^W||_1 - d$.
\end{lemma}

\begin{proof}
There are two cases.
First, suppose $C = N_G(W) = \{u\}$ (the case $D = N_G(U) = \{w\}$ is symmetrical to this one).
The split marker vertex $w$ is pendant in $G_U$, with its unique neighbour being $u$ (so, in particular, $x^U_{\{u,w\}} = d$).
In order to compute $x$, we set $x^U_{\{u,w\}}$ to $0$ and then we merge $x^U,x^W$.
By Lemma~\ref{lem:data-structure} this takes constant-time.
Furthermore, since $|E| - |E_U| - |E_W| = |C||D| - |C| - |D| = -1$, and $sc(G) = sc(G_U) + sc(G_W)$, we get $\varphi(G) - \varphi(G_U) - \varphi(G_W) = 3 > 0$.

Therefore, from now on we assume that $|C| \geq 2$ and $|D| \geq 2$.
For every $v \in C$ we assign a capacity $c_v = x^U_{\{v,w\}}$ and then we set $x^U_{\{v,w\}}$ to $0$.
In the same way, for every $v' \in D$ we assign a capacity $c_{v'} = x^W_{\{v',u\}}$ and then we set $x^W_{\{v',u\}}$ to $0$.
It takes ${\cal O}(|C| + |D|)$-time.
Then, let us merge $x^U,x^W$ in order to initialize $x$.
By Lemma~\ref{lem:data-structure} this takes constant-time.
While there exist a $v \in C$ and a $v' \in D$ such that $c_v > 0, \ c_{v'} > 0$ we pick one such pair $v,v'$ and we set: $x_{\{v,v'\}} = \min\{c_v,c_{v'}\}; \ c_v = c_v - x_{\{v,v'\}}; \ c_{v'} = c_{v'} - x_{\{v,v'\}}$.
Since for every loop, the capacity of at least one vertex drops to $0$, it takes total time ${\cal O}(|C| + |D|)$.
Furthermore, since $|C| \geq 2$ and $|D| \geq 2$ we have $|E| - |E_U| - |E_W| = |C||D| - (|C| + |D|) \geq 2(|C| + |D|) - 4 - (|C| + |D|) \geq |C| + |D| - 4$.
As a result, $\varphi(G) - \varphi(G_U) - \varphi(G_W) \geq |C|+|D| - 4 + 4 \geq \Omega(|C|+|D|)$.
\end{proof}

Overall, since there are at most $n-2$ components in any minimal split decomposition of $G$~\cite{Rao08b}, the merging stages take total time ${\cal O}(\varphi(G)) = {\cal O}(n+m)$.

\subsection*{Main result}

Our main result can be seen as an algorithmic proof of the following equality:

\begin{proposition}\label{prop:reduction-rule}
Let $G=(V,E),b$, let $(U,W)$ be a split of $G$ and let $H_W,b^W$ be as defined in Definition~\ref{def:reduction-rule}.
Given a maximum-cardinality $b$-matching $x$ for the pair $G,b$ and a maximum-cardinality $b$-matching $x^W$ for the pair $H_W,b^W$ we have: $$||x||_1 = ||x^W||_1 + \mu^U(0) - c_2^U $$
\end{proposition}

\begin{proof}

We have $||x^W||_1 \geq ||x||_1 - \mu^U(0) + c_2^U$ by Lemma~\ref{lem:reduction-rule-1}.
In order to prove the converse inequality, we can assume w.l.o.g. that $x^W$ satisfies both the saturation property and the symmetry property w.r.t. the module $M_u$ (otherwise, by Lemma~\ref{lem:simultaneous}, we can process $x^W$ so that it is the case).
We partition $||x^W||_1$ as follows: $\mu^W = \sum_{e \in E(W)} x^W_e$, $c_1' = deg_{x^W}(u_1) \leq c_1^U$ and $c_2' = deg_{x^W}(u_2) - x^W_{\{u_2,u_3\}} = deg_{x^W}(u_3)- x^W_{\{u_2,u_3\}} \leq c_2^U$.
Since we assume that $x^W$ satisfies both the saturation property and the symmetry property w.r.t. $M_u$, we have $c_2' > 0$ only if $c_1' = c_1^U$.
Furthermore, we observe that $u_2$ and $u_3$ must be saturated (otherwise, we could increase the cardinality of the $b$-matching by setting $x_{\{u_2,u_3\}}^W = c_2^U - c_2'$).
Therefore, we get: $$||x^W||_1 = \mu^W + c_1' + 2c_2' + (c_2^U - c_2') = \mu^W + c_1' +c_2' + c_2^U.$$

\begin{figure}[h!]
\centering
\includegraphics[width=.8\textwidth]{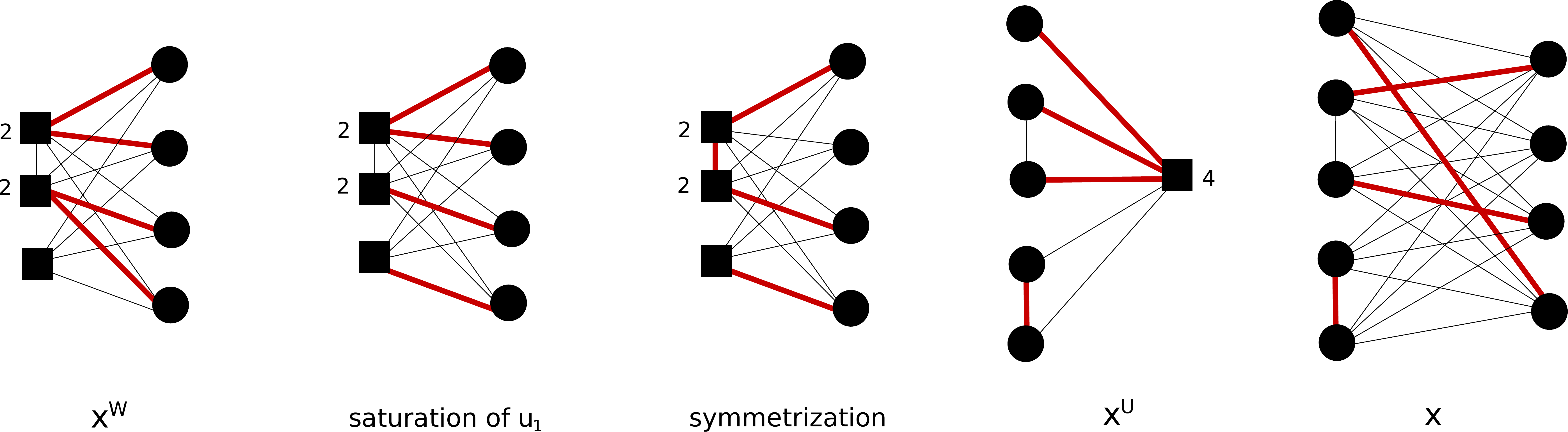}
\caption{The construction of $x'$. Vertices with capacity greater than $1$ are labeled with their capacity. Thin and bold edges have respective weights $0$ and $1$.}
\label{fig:reduction-2}
\end{figure}

We define $b_u^W = b_w^U = c_1' + 2c_2'$.
Then, we proceed as follows (see Fig.~\ref{fig:reduction-2} for an illustration).
\begin{itemize}
\item
We transform $x^W$ into a $b$-matching for the pair $G_W,b^W$ by setting $x^W_{\{u,v'\}} = x^W_{\{u_1,v'\}} + x^W_{\{u_2,v'\}} + x^W_{\{u_3,v'\}}$ for every $v' \in N_{G_W}(u) = D$.
Note that we have $deg_{x^W}(u) = b_u^W = c_1' + 2c_2'$.
Furthermore, the cardinality of the $b$-matching has decreased by $x^W_{\{u_2,u_3\}} = c_2^U - c_2'$.
\item
Let $x^U$ be a $b$-matching for the pair $G_U,b^U$ of maximum cardinality $\mu^U(c_1'+2c_2')$.
Since $c_1' \leq c_1^U$, $c_2' > 0$ only if $c_1' = c_1^U$, and $c_2' \leq c_2^U$, the following can be deduced from Proposition~\ref{prop:monotonic}: $||x^U||_1 = \mu^U(c_1'+2c_2') = \mu^U(0) + c_1' + c_2'$; and the split marker vertex $w$ is saturated in $x^U$, {\it i.e.}, $deg_{x^U}(w) = b^U_w = c_1'+2c_2'$.
\end{itemize}
Finally, since we have $deg_{x^W}(u) = deg_{x^U}(w) = c_1'+2c_2'$, we can define a $b$-matching $x'$ for the pair $G,b$ by applying Lemma~\ref{lem:potential-function}. 
Doing so, we get $||x||_1 \geq ||x'||_1 = ||x^U||_1 + \left(||x^W||_1 - (c_2^U - c_2')\right) - (c_1'+2c_2') = \mu^U(0) + c_1' + c_2' + ||x^W||_1 - (c_2^U + c_1'+c_2') = ||x^W||_1 + \mu^U(0) - c_2^U$.
\end{proof}

Finally, we now prove the main result in this paper.

\begin{theorem}\label{thm:main-result}
For every pair $G=(V,E), b$ with $sw(G) \leq k$, we can solve $b$-{\sc Matching} in ${\cal O}((k\log^2{k})\cdot(m+n) \cdot \log{||b||_1})$-time.
\end{theorem}

\begin{proof}
Let $C_1, C_2, \ldots, C_s, \ s = sc(G)$ be the split components in any minimal split decomposition of $G$.
Furthermore, let $T$ be the corresponding split decomposition tree.
It can be computed in linear-time~\cite{Rao08b}.
We root $T$ in $C_1$.

For every $1 \leq i \leq s$, let $T_i$ be the subtree of $T$ that is rooted in $C_i$.
If $i > 1$ then let $C_{p(i)}$ be its parent in $T$.
By construction of $T$, the edge $\{C_{p(i)}, C_i\} \in E(T)$ corresponds to a split $(U_i,W_i)$ of $G$, where $V(C_i) \cap V \subseteq U_i$.
Let $G_{U_i} = (U_i \cup \{w_i\}, E_{U_i}), \ G_{W_i} = (W_i \cup \{u_i\}, E_{W_i})$ be the corresponding subgraphs of $G$.
As we observed in~\cite{CDP18}, $T_i$ is a split decomposition tree of $G_{U_i}$, $T \setminus T_i$ is a split decomposition tree of $G_{W_i}$.

\medskip
Our algorithm proceeds in two main steps, with each step corresponding to a different traversal of the tree $T$.
First, let $G_1 = G$ and let $G_i = G_{U_i}$ for every $i > 1$.
Let $b^1 = b$ and, for every $i > 1$ let $b^i$ be the restriction of $b$ to $U_i$.
We note that for any $i > 1$, $b^i$ does not assign any capacity to the split marker vertex $w_i$.
Up to adding a dummy isolated vertex $w_1$ to $G_1$, we assume that this above property holds for any $i$.
Then, for any $i$, we compute the triple $(\mu^i(0),c_1^i,c_2^i)$ w.r.t. $G_i,b^i$ and $w_i$ (as defined in Proposition~\ref{prop:monotonic}).

In order to do so, we proceed by dynamic programming on the tree $T$, as follows.
Let $C_{i_1}, C_{i_2}, \ldots, C_{i_l}$ be the children of $C_i$ in $T$.
Every edge $\{C_i, C_{i_t}\} \in E(T), \ 1 \leq t \leq l$ corresponds to a split $(U_{i_t},W_{i_t})$ of $G_i$, where $V(C_{i_t}) \cap V(G_i) \subseteq U_{i_t}$.
We name $w_{i_t} \in V(C_{i_t}), \ u_{i_t} \in V(C_i)$ the vertices added after the simple decomposition.
Furthermore, let $(\mu^{i_t}(0),c_1^{i_t},c_2^{i_t})$ be the triple corresponding to $G_{i_t},b^{i_t}$ and $w_{i_t}$ (obtained by dynamic programming).
We apply the reduction rule of Definition~\ref{def:reduction-rule} --- {\it i.e.}, we replace the split marker vertex $u_{i_t}$ by the module $M_{i_t} = \{u_{i_t}^1,u_{i_t}^2,u_{i_t}^3\}$ where $u_{i_t}^1$ has capacity $c_1^{i_t}$ and the two adjacent vertices $u_{i_t}^2,u_{i_t}^3$ have equal capacity $c_2^{i_t}$.
Doing so for every $1 \leq t \leq l$, we obtain a pair $H_i,b^{H_i}$ where $H_i$ is a supergraph of $C_i$ such that: $|V(H_i)| \leq 3|V(C_i)| \leq \max\{3k,9\}$ and $|E(H_i)| \leq 9|E(C_i)| + |V(C_i)|$.
Let us compute the triple $(\mu^{H_i}(0),c_1^i,c_2^i)$ corresponding to $H_i,b^{H_i}$ and $w_i$.
By Proposition~\ref{prop:monotonic} it can be done in time ${\cal O}(|V(H_i)||E(H_i)|\log^2{|V(H_i)|}\log{||b||_1})$, that is in ${\cal O}((k\log^2{k})\cdot(|E(C_i)| + |V(C_i)|)\log{||b||_1})$.

Finally, by applying Proposition~\ref{prop:reduction-rule} for every split $(U_{i_t},W_{i_t})$ we have: $$\mu^i(0) = \mu^{H_i}(0) + \sum_{t=1}^l (\mu^{i_t}(0) - c_2^{i_t}).$$
Overall, this step takes total time ${\cal O}((k\log^2{k}) \cdot \sum_i (|E(C_i)| + |V(C_i)|) \cdot \log{||b||_1}) = {\cal O}((k\log^2{k}) \cdot (n + m) \cdot \log{||b||_1})$.
Furthermore, since $G_1 = G$, we have computed the maximum cardinality $\mu^1(0)$ of any $b$-matching of $G$.

\medskip
Second, we compute a $b$-matching for the pair $G,b$ that is of maximum cardinality $\mu^1(0)$, by reverse dynamic programming on $T$.
More precisely, for any $i$ let $b_{w_i}^i$ be a fixed capacity for the split marker vertex $w_i$ (if $i=1$ then we set $b_{w_1}^1 = 0$; otherwise, for $i >1$, $b_{w_i}^i$ is obtained by reverse dynamic programming).
In what follows, we compute a maximum-cardinality $b$-matching for the pair $G_i,b^i$.
For that we set $b^{H_i}_{w_i} = b^i_{w_i}$.
\begin{itemize}
\item
We compute a maximum-cardinality $x^i$ for the pair $H_i,b^{H_i}$ such that both the saturation property and the symmetry property hold for every $1 \leq t \leq l$.
By Lemma~\ref{lem:simultaneous}, it can be done in time ${\cal O}(|V(H_i)||E(H_i)|\log^2{|V(H_i)|})$, that is in ${\cal O}((k\log^2{k})\cdot(|E(C_i)| + |V(C_i)|))$.
\item
For every $1 \leq t \leq l$, let us define $b^i_{u_{i_t}} = deg_{x^i}(u_{i_t}^1) + 2 \cdot deg_{x^i}(u_{i_t}^2)$.
We merge $M_{i_t}$ into the original split marker vertex $u_{i_t}$.
Furthermore, we assign the capacity $b^i_{u_{i_t}}$ to $u_{i_t}$, and we update the $b$-matching $x^i$ such that, for every $v \in N_{H_i}(M_{i_t})$, we have $x^i_{\{v,u_{i_t}\}} = x^i_{\{v,u_{i_t}^1\}} + x^i_{\{v,u_{i_t}^2\}} + x^i_{\{v,u_{i_t}^3\}}$.
Doing so, we transform $x^i$ into a $b$-matching for $C_i,b^i$ where all vertices $u_{i_1},u_{i_2},\ldots,u_{i_l}$ are saturated.
This transformation has decreased the cardinality of the solution by: $$\sum_{t=1}^l x^i_{\{u_{i_t}^2,u_{i_t}^3\}} \leq \sum_{t=1}^l c_2^{i_t}.$$
\item
For every $1 \leq t \leq l$ we set $b^{i_t}_{w_{i_t}} = b^i_{u_{i_t}}$ and then we compute a maximum-cardinality $b$-matching $x^{i_t}$ for the pair $G_{i_t},b^{i_t}$.
Since the symmetry and saturation properties hold, we can deduce from Proposition~\ref{prop:monotonic} that $w_{i_t}$ is saturated.
\item
By applying the routine of Lemma~\ref{lem:potential-function} for every incident split $(U_{i_t},W_{i_t})$ we merge $x^i$ with all the $x^{i_t}$'s until we obtain a $b$-matching for $G_i,b^i$.
It takes time ${\cal O}(\varphi(G_i) - \sum_{t=1}^l \varphi(G_{i_t}))$.
\end{itemize}
Finally, the above $b$-matching is of maximum cardinality, that follows from Proposition~\ref{prop:reduction-rule} (applied for every incident split).

Overall, this second step of the algorithm takes total time ${\cal O}(\varphi(G)) + {\cal O}((k\log^2{k}) \cdot \sum_i (|E(C_i)| + |V(C_i)|)) = {\cal O}((k\log^2{k}) \cdot (n + m))$.
\end{proof}

\begin{corollary}\label{cor:max-matching}
For every graph $G=(V,E)$ with $sw(G) \leq k$, we can solve {\sc Maximum Matching} in ${\cal O}((k\log^2{k})\cdot(m+n) \cdot \log{n})$-time.
\end{corollary}

\section{Open questions}\label{sec:ccl}

We presented an algorithm for solving $b$-{\sc Matching} on distance-hereditary graphs, and more generally on any graph with bounded split-width.
It would be interesting to extend our technique in order to solve the more general {\sc Maximum-Cost} $b$-{\sc Matching} problem.
In particular, this would imply a quasi linear-time algorithm for {\sc Maximum-Weight Matching} on bounded split-width graphs.
Our main difficulty for handling with this more general case is that edges in a different split can now have pairwise different costs.
We do not know how to encode this additional information within our gadgets in a compact manner.

Finally, we stress that, as a byproduct of our main result, we obtained an ${\cal O}((n+m)\log n)$-time {\sc Maximum Matching} algorithm on graphs with bounded split-width. 
We ask whether there exists a linear-time algorithm for this problem.
In a companion paper~\cite{DuP18+}, we prove with a completely different approach that {\sc Maximum Matching} can be solved in ${\cal O}(n+m)$-time on distance-hereditary graphs.
However, it is not clear to us whether similar techniques can be used for bounded split-width graphs in general.

\bibliographystyle{abbrv}
\small
\bibliography{bibliography-split}

\end{document}